\documentclass[12pt]{article}
\usepackage{amsmath}
\usepackage{graphicx,psfrag,epsf}
\usepackage{amssymb, amsmath, amsthm, color,enumerate,float}
\usepackage{algorithm, comment,subfig}
\usepackage{enumerate,multirow,booktabs}
\usepackage{url}
\usepackage{algpseudocode}
\usepackage{natbib}
\usepackage{array}  
\usepackage{url} 
\usepackage[cal=cm]{mathalfa}
\usepackage{authblk}

 \advance\headheight by 2pt
\advance\textheight by-\headheight \advance\textheight by-\headsep
\graphicspath{{figures/}}
\usepackage{url}

\newtheorem{thm}{Theorem}

\newtheorem{cor}{Corollary}
 \DeclareMathOperator{\var}{var}
\DeclareMathOperator{\tr}{tr} 
\DeclareMathOperator{\IMSE}{IMSE}
\newcommand{\E}{\mathbb{E}}
\DeclareMathOperator{\diag}{diag}
\DeclareMathOperator*{\argmin}{argmin}

\DeclareMathOperator{\eff}{eff}
\DeclareMathOperator{\opt}{opt}
\newcommand{\reals}{\mathbb{R}}
\newcommand{\naturals}{\mathbb{R}}
\newcommand{\T}{\intercal}
\newcommand{\vx}{\boldsymbol{x}}

\DeclareMathOperator{\EI}{EI}

\newcommand{\vt}{\boldsymbol{t}}
\newcommand{\vz}{\boldsymbol{z}}

\newcommand{\vg}{\boldsymbol{g}}

\newcommand{\unif}{\textup{unif}}
\newcommand{\asin}{\textup{asin}}

\newcommand{\vbeta}{\boldsymbol{\beta}}

\newcommand{\mA}{{\mathsf A}}
\newcommand{\mL}{{\mathsf L}}

\newcommand{\mI}{{\mathsf I}}

\newcommand{\valpha}{\boldsymbol{\alpha}}

\newcommand{\dif}{{\rm d}}

\newcommand{\ch}{\mathcal{H}}

\newcommand{\Linear}{{\text L}}

\newcommand{\tar}{\text{\rm tar}}

\def\abs#1{\ensuremath{\left \lvert #1 \right \rvert}}

\begin{document}

\title{On Efficient Design of Pilot Experiment for Generalized Linear Models}

\author[1]{Yiou Li}
\author[2]{Xinwei Deng\thanks{Address for correspondence: Xinwei Deng, Associate Professor, Department of Statistics, Virginia Tech, Blacksburg, VA, 24061 (xdeng@vt.edu).}}
\affil[1]{Department of Mathematical Sciences, DePaul University}
\affil[2]{Department of Statistics, Virginia Tech}

\date{}
\maketitle

\begin{abstract}
The experimental design for a generalized linear model (GLM) is important but challenging since the design criterion often depends on model specification including the link function, the linear predictor, and the unknown regression coefficients.
Prior to constructing locally or globally optimal designs, a pilot experiment is usually conducted to provide some insights on the model specifications.
In pilot experiments, little information on the model specification of GLM is available.
Surprisingly, there is very limited research on the design of pilot experiments for GLMs.
In this work, we obtain some theoretical understanding of the design efficiency in pilot experiments for GLMs.
Guided by the theory, we propose to adopt a low-discrepancy design with respect to some target distribution for pilot experiments.
The performance of the proposed design is assessed through several numerical examples.

\end{abstract}

\section{Introduction}\label{sec:intro}
Various experimental design problems encounter the non-normal response such as the binary outcome and the number of events \citep{wu2011experiments}.
While the generalized linear models (GLMs) \citep{72nelder} are commonly used to analyze the data with non-normal responses, the experimental design issues for GLMs are challenging since the design criterion often relies on model specification including the link function, the linear predictor, and the unknown regression coefficients.
For a generalized linear model, let us assume that the $d$-dimensional design variable $\vx = [x_1,\ldots,x_d]^{\top}$ is drawn from some experimental region $\Omega$.
The experimental region $\Omega$ could be bounded or unbounded, such as $[-1,1]^d$ or $\mathbb{R}^d$ .
The response variable $Y(\vx)$ of a GLM is considered to follow a distribution in the exponential family.
The mean response $\mu(\vx)$ is related to the design variable $\vx$ through a link function $h$,
$$\mu(\vx) = \E[Y(\vx)] = h^{-1}(\eta(\vx)),$$
where $\vg = [g_1,\ldots,g_l]^{\top}$ is the basis function, $\vbeta = [\beta_1,\ldots,\beta_l]^{\top}$ is the regression coefficient,  $\eta(\vx) = \vbeta^{\top}\vg(\vx)$ is the linear predictor, and $h^{-1}$ is the inverse function of $h$.
Allowing repeated measurements, consider an experiment with $n$ observations at $m$ distinct design points, and the corresponding exact design could be expressed as:
\begin{equation}\label{eqn:design}
\xi =  \left\{\begin{array}{ccc}
\vx_1, & \ldots, & \vx_m\\
n_1, & \ldots, & n_m
\end{array}\right\},
\end{equation}
where $n_i$ is the number of repetitions at design point $\vx_i$, and $\sum\limits_{i=1}^m n_i=n$, the size of the design.  Denote the empirical distribution of design $\xi$ as $F_{\xi}$. For a model specification $M = (h,\vg,\vbeta)$, the information matrix of an exact design $\xi$ is
\begin{equation}\label{eqn:fisher}
\mI(\xi;M)=\sum\limits_{i=1}^m \frac{n_i}{n}\vg(\vx_i)w(\vx_i;M)\vg^\top(\vx_i),
\end{equation}
where $w(\vx_i;M) = \left[\var(Y(\vx_i))[h^{'}(\mu(\vx_i))]^2\right]^{-1}$.
Clearly, the design issue for GLMs is complicated and challenging due to the dependence of the information matrix $\mI(\xi;M)$ on all elements of the model specification $M$.

The experimental designs for GLMs have been extensively studied, under the assumption that the model space $\mathcal{M}$ containing all model specifications of interest is available from a pilot experiment.
In such pilot experiments, there is little information about the model specification.
Actually, the main purpose of a pilot experiment is to obtain valuable information on the choice of appropriate link and basis functions, and consequently, obtain some initial estimate of the regression coefficients.
A general and flexible design criterion to assess the accuracy of coefficient estimates is L-optimality,
which aims at minimizing $\L_{\opt}(\xi;M) = \tr\left[\mI^{-1}(\xi;M)\mL\right]$
with $\mL$ to be an $l\times q$ matrix.
When the rank of $\mL$ is 1, L-optimality becomes c-optimality that minimizes a linear combination of the variances of the coefficient estimates.
When $\mL$ is chosen to be the identity matrix, L-optimality becomes the classical A-optimality that minimizes the total variance of the coefficient estimates.
The \emph{`standardized'} A-optimality $SA_{\text{opt}}(\xi;M) =\sum_{j=1}^l \frac{\left(\mI^{-1}(\xi;M)\right)_{jj}}{\left(\mI^{-1}(\xi_j^*;M)\right)_{jj}} = \tr\left[\mI^{-1}(\xi;M)\mL\right]$ proposed by \cite{dette1997designing} is also a special case of L-optimality, where $\xi_j^* = \argmin_{\xi} (\mI^{-1}(\xi;M))_{jj}$ is the design that minimizes the asymptotic variance of the maximum likelihood estimator $\hat{\beta}_j$ for the $j$th coefficient $\beta_j$, and $ \mL = \diag\left[(1/\left(\mI^{-1}(\xi_j^*;M)\right)_{jj})_{j=1}^l\right]$.
Compared to the  classical A-optimality \citep{72fed, atkinson2007optimum},  the `standardized' A-optimality takes into consideration that the variances of coefficient estimators could be of different scales.
A scale-free measure to assess the performance of a design under L-optimality is L-efficiency,
\begin{equation}
\eff_{\Linear}(\xi,\xi^{\opt}_M;M) = \frac{\Linear_{\text{opt}}(\xi^{\opt}_M;M)}{\Linear_{\text{opt}}(\xi;M)},
\end{equation}
where $\xi^{\opt}_M = \argmin_{\xi} \Linear_{\text{opt}}(\xi;M)$ is the locally L-optimal design for model specification $M = (h,\vg,\vbeta)$.  Obviously, $0\leq \eff_{\Linear}(\xi,\xi^{\opt}_M;M)\leq 1$ for any design $\xi$, and the larger the L-efficiency $\eff_{\Linear}(\xi,\xi^{\opt}_M;M)$, the more efficient the design $\xi$ is.

%
After obtaining some preliminary understanding of the model specifications from the pilot experiment,  various locally or globally optimal designs can be constructed, including \citep{imhof2000graphical, amzal2006bayesian, tekle2008maximin,  woods2011continuous, 13yang,  dean2015handbook, woods2017bayesian, li2020efficient}, among many others.
While little knowledge of the model specification is available in a pilot experiment,
it calls for a flexible and efficient design, which can regulate the L-efficiency for all model specifications $M$ in a model space $\mathcal{M}$ containing a wide variety of model specifications.

In practice, fractional factorial designs and space-filling designs are the common choices for the pilot experiments to obtain some initial understanding of the GLMs. There are two main drawbacks of  the fractional factorial designs. One is the fixed design size, and as a result, the design size grows exponentially as the design variable dimension $d$ gets large. Secondly, the number of levels depends on the basis functions. For instance, a two-level fractional factorial design can not be applied when the basis functions contain quadratic terms.
To our best knowledge, the literature on the design of pilot experiments is surprisingly scarce.
In this work, we establish a tight lower bound of  L-efficiency to investigate the designs of pilot experiments for GLMs.
This lower bound provides a theoretical rationale for seeking efficient and robust designs for pilot experiments of GLMs.
Guided by the theoretical result, we propose to use the discrepancy with respect to some target distribution as the design criterion, which is robust against the unknown model specification, to regulate the L-efficiency of the design for pilot experiments.
The proposed design criterion and corresponding design require very mild assumptions on the model space $\mathcal{M}$ and hence is suitable for pilot experiments of GLMs.

The rest of the work is organized as follows.
In Section \ref{sec:lowdiscrepancydesign},  the design criterion, discrepancy, that measures the difference between the empirical distribution of a design and a target distribution $F_{\tar}$,  is introduced. By deriving a tight lower bound of L-efficiency of a design for all model specifications in a model space that requires little model assumptions, the theoretical rationale of the proposed design criterion is justified.
The discrepancies of commonly used designs in the literature  and their L-efficiency performance are assessed through numerical examples in Section \ref{sec:examples}.
We conclude this work with some discussions in Section \ref{sec:conclusion}.


\section{Low Discrepancy Design and Its L-Efficiency}\label{sec:lowdiscrepancydesign}
Denote the model space,  i.e., the set of all possible model specifications,  to be $\mathcal{M}$.
Apparently,  before conducting the pilot experiments, there is little information about the model space $\mathcal{M}$, and a key objective of the pilot experiments is to obtain some valuable information about model space $\mathcal{M}$ so that a locally/globally optimal design  can be further constructed.
Therefore, in the pilot experiment, the experimenter would prefer a design that is robust and efficient, in other words, guarantees a reasonably large L-efficiency, over a wide class of model specifications.

In this section,  we first introduce a design criterion, discrepancy, as a measure of the difference between the empirical distribution of a design $\xi$ and some continuous target distribution $F_{\tar}$. We then describe the reproducing kernel Hilbert space (RKHS) that defines the model space $\mathcal{M}$, which accounts for a high level of model uncertainty with an appropriate choice of reproducing kernel. By deriving a tight lower bound of L-efficiency for all model specifications in the model space $\mathcal{M}$ defined above, we show that a design with a small discrepancy would be an appropriate design that regulates the L-efficiency over a large variety of model specifications for a pilot experiment.

\subsection{Discrepancy Measures the Difference between $\xi$ and $F_{\tar}$}
Consider some target distribution with cumulative distribution function $F_{\tar}$,  which lies in a space $\mathcal{F}$ of signed measures defined by some reproducing kernel $K$.  A reproducing kernel $K$ is a symmetric and semi-positive definite function that satisfies:
\begin{subequations}\label{eqn:kernelconditions}
\begin{gather}
K(\vx,\vt)=K(\vt,\vx)\,\,\,\forall \vt, \vx \in\Omega, \\
\sum_{i,k}c_i c_j K(\vx_i,\vx_j)\geq 0\qquad \forall n \in \naturals, \ c_1, \ldots, c_n \in\reals, \ \vx_1, \ldots, \vx_n \in\Omega.
\end{gather}
\end{subequations}
For a kernel $K$, the measure space $\mathcal{F}$ consists of signed measures
$$\mathcal{F} = \left\{\text{signed measure } F(x): \int_{\Omega^2} K(\vx, \vt)\dif F(\vx)\dif F(\vt)<\infty\right\},$$
and is equipped with the inner product
$$\langle F, G\rangle_{\mathcal{F}} = \int_{\Omega^2} K(\vx, \vt)\dif F(\vx)\dif G(\vt).$$
Under the context of experimental design, an exact design
\begin{equation*}
\xi =  \left\{\begin{array}{ccc}
\vx_1, & \ldots, & \vx_m\\
n_1, & \ldots, & n_m
\end{array}\right\}
\end{equation*}
induces a corresponding empirical distribution $F_\xi$ to approximate the target distribution $F_{\tar}$ on $\Omega$. Then, the distance between $F_{\xi}$ and the target distribution $F_{\tar}$ is called the \emph{discrepancy} between $F_{\xi}$ and $F_{\tar}$ \citep{hickernell1999goodness}, i.e.,
\begin{align}
\nonumber
D(\xi;F_{\tar}) &= \|F_{\xi}-F_{\tar}\|_{\mathcal{F}}\nonumber \\
&=\left[\int_{\Omega^2}K(\vx,\vt) \, \dif \{F_{\tar}(\vx)-F_\xi(\vx)\} \dif \{F_{\tar}(\vt)-F_\xi(\vt)\}\right]^{\frac{1}{2}}\nonumber\\
&=
\left[\int_{\Omega^2}K(\vx,\vt) \, \dif F_{\tar}(\vx) \dif F_{\tar}(\vt)-\frac{2}n \sum_{i=1}^m n_i\int_{\Omega}K(\vx_i,\vt) \, \dif F_{\tar}(\vt)\right.\nonumber\\
&\qquad \qquad
\left.+ \frac{1}{n^2}\sum\limits_{i,k=1}^m n_in_kK(\vx_i,\vx_k)\right]^{\frac{1}{2}}.
\label{discdef}
\end{align}
Although $D(\xi;F_{\tar})$ depends on the kernel $K$,
this dependence is suppressed for simplicity of notation.

\subsection{Discrepancy Design Criterion to Regulate L-Efficiency}\label{sec:boundleff}
With any reproducing kernel $K$, we can also define a reproducing kernel Hilbert space (RKHS) $\ch$, which is a separable Hilbert space \citep{Aro50, BerT-A04, Fas07a}.   $K(\cdot,\vx)$ is the representer for the functional
that evaluates a function at a point $\vx$, that is,
$$
K(\cdot, \vx)\in\mathcal{H},\quad f(\vx)=\langle f, K(\cdot,
\vx)\rangle_{\mathcal{H}} \qquad \forall
\vx\in\Omega,f\in\mathcal{H}.$$
Moreover, any function satisfying two conditions in \eqref{eqn:kernelconditions} is the
reproducing kernel for some unique Hilbert space. In numerical analysis, the reproducing kernel Hilbert spaces are commonly used to provide a tight upper bound on numerical integration error.
It is known that \citep{hickernell1999goodness}, for any $f\in\mathcal{H}$,
\begin{align}
\nonumber
\abs{\int_{\Omega}f(\vx)dF_{\tar}(\vx)-\sum_{i=1}^m \frac{n_i}{n}f(\vx_i)} & =\abs{\int_{\Omega}f(\vx) \, \dif [F_{\tar}(\vx)-F_\xi(\vx)]} \\
\label{eqn:KH}
& \le D(\xi;F_{\tar}) V(f),
\end{align}
where $D(\xi;F_{\tar})$ is the \emph{discrepancy} defined in \eqref{discdef} to measure the difference between $F_{\xi}$ and $F_{\tar}$ by the kernel $K$. Moreover, the \emph{variation} of the integrand, $V(f)$ in \eqref{eqn:KH}, which measures the roughness/degrees of oscillation of $f$, is the (semi-) norm of its non-constant part:
\begin{equation} \label{vardef}
V(f)=\begin{cases} \|f\|_{\mathcal{H}}&\text{if} \,\,1\notin\mathcal{H},\\
\left(\|f\|^2_{\mathcal{H}}-\langle
f,1\rangle^2_{\mathcal{H}}/\|1\|^2_{\mathcal{H}}\right)^{1/2}&\text{if}\,\,
1\in\mathcal{H}.\end{cases}
\end{equation}

Using the discrepancy bound \eqref{eqn:KH}, in what follows we will make a connection between the discrepancy and L-efficiency of a design. Consequently,  we will show that a design with a small discrepancy can regulate the L-efficiency for all model specifications in a model space $\mathcal{M}$ defined by the RKHS. With an appropriate choice of the reproducing kernel $K$, the corresponding model space $\mathcal{M}$ requires very mild model assumptions and contains a wide variety of model specifications.
Before stating the theorem, we need to define a few notations.
For any model specification $M = (h,\vg,\vbeta)$, define the Fisher information matrix for a design of  a single point $\vx$ with unit mass as $\mI_{\vx} = \vg(\vx)w(\vx;M)\vg^\top(\vx)$.
Then the Fisher information matrix for a continuous design with target distribution $F_{\tar}$ on $\Omega$ is expressed as
$$\mI({\tar};M)=\int_\Omega \vg(\vx)w(\vx;M)\vg^\top(\vx)\dif F_{\tar} = \int_{\Omega} \mI_{\vx}\dif F_{\tar},$$
and the Fisher information matrix for an exact design $\xi$ in \eqref{eqn:fisher} can be expressed as
$\mI(\xi;M)=\int_\Omega \mI_{\vx}\dif F_{\xi}.$

The result in Theorem \ref{thm:robust} below provides a tight lower bound of L-efficiency of a design $\xi$ for all model specifications $M\in\mathcal{M}$.
This lower bound implies that a design with a small discrepancy can be robust and efficient to a wide variety of model specifications with an appropriate choice of kernel $K$ and target distribution $F_{\tar}$.
\begin{thm}\label{thm:robust}
Suppose that $\ch$ is a reproducing kernel Hilbert space of functions defined on $\Omega$ with kernel $K$.
Assume that the function $u_{\valpha;M}: \vx\mapsto \valpha^\T[\mI({\tar};M)]^{-\frac{1}{2}}\mI_{\vx}[\mI({\tar};M)]^{-\frac{1}{2}}\valpha$ lies in $\ch$ for any $\valpha\in \reals^l$ and any $M\in\mathcal{M}$.
Define the variation over a model specification $M$ as
$$V_{M}=\sup_{\|\valpha\|_2\leq 1} V(u_{\valpha;M}),$$
where variation $V(\cdot)$ is defined in \eqref{vardef}.
Then it follows that for a design $\xi$, the L-efficiency of $\xi$ for any model specification $M\in\mathcal{M}$ is tightly bounded below by
\begin{eqnarray}\label{lowerbound}
\eff_{\Linear}(\xi,\xi^{\opt};M)\geq [1-D(\xi;F_{\tar})\cdot V_{M}]\eff_{\Linear}(\tar,\xi^{\opt};M),\,\,\,\,\forall M\in\mathcal{M},
\end{eqnarray}
where $\eff_{\Linear}(\tar,\xi^{\opt};M) = \frac{\Linear_{\opt}(\xi^{\opt};M)}{\Linear_{\opt}(\tar;M)}$.
\end{thm}

\begin{proof}
Define
\begin{eqnarray*}
\tilde{\mI} &=& \mathbcal{I}_{l\times l} - [\mI({\tar};M)]^{-1/2}\mI(\xi;M)[\mI({\tar};M)]^{-1/2}\\
& =& [\mI({\tar};M)]^{-1/2}\left(\mI({\tar};M)-\mI(\xi;M)\right)[\mI({\tar};M)]^{-1/2},
\end{eqnarray*}
where $\mathbcal{I}_{l\times l}$ is a $l\times l$ identity matrix. Then, the spectral radius of $\tilde{\mI}$ is
\begin{eqnarray*}
\rho(\tilde{\mI}) &=& \sup_{\|\valpha\|_2\leq 1} |\valpha^\T\tilde{\mI}\valpha|\\
&=& \sup_{\|\valpha\|_2\leq 1} \left|\valpha^\T[\mI({\tar};M)]^{-1/2}\left(\mI({\tar};M)-\mI(\xi;M)\right)[\mI({\tar};M)]^{-1/2}\valpha\right|\\
&=& \sup_{\|\valpha\|_2\leq 1} \left|\valpha^\T[\mI({\tar};M)]^{-1/2}\left(\int_{\Omega}\mI_{\vx}\dif \left[F_{\tar}-F_{\xi}\right]\right)[\mI({\tar};M)]^{-1/2}\valpha\right|\\
&=& \sup_{\|\valpha\|_2\leq 1} \left|\int_{\Omega}\valpha^\T[\mI({\tar};M)]^{-1/2}\mI_{\vx}[\mI({\tar};M)]^{-1/2}\valpha\dif \left[F_{\tar}-F_{\xi}\right]\right|\\
&\leq & D(\xi;F_{\tar})\sup_{\|\valpha\|_2\leq 1}V(\valpha^\T[\mI({\tar};M)]^{-1/2}\mI_{\vx}[\mI({\tar};M)]^{-1/2}\valpha) = D(\xi;F_{\tar}) \cdot V_{M},
\end{eqnarray*}
where the last inequality comes from \eqref{eqn:KH}. Note that $(1-\rho(\tilde{\mI}))$ is the smallest eigenvalue of $[\mI({\tar};M)]^{-1/2}\mI(\xi;M)[\mI({\tar};M)]^{-1/2}$.

When $D(\xi;F_{\tar})\cdot V_{M}<1$, the smallest eigenvalue of $[\mI({\tar};M)]^{-1/2}\mI(\xi;M)[\mI({\tar};M)]^{-1/2}$ is no smaller than $1-D(\xi;F_{\tar})\cdot V_{M}$.
Since $[\mI({\tar};M)]^{-1/2}\mI(\xi;M)[\mI({\tar};M)]^{-1/2}$ is a positive definite matrix, the largest eigenvalue of its inverse $\left\{[\mI({\tar};M)]^{-1/2}\mI(\xi;M)[\mI({\tar};M)]^{-1/2}\right\}^{-1}$ is no larger than $1/(1-D(\xi;F_{\tar})\cdot V_{M})$, provided $D(\xi;F_{\tar})\cdot V_{M}<1$.

Then, for any model specification $M\in\mathcal{M}$,  the L-efficiency of a design $\xi$, $\eff_{\Linear}(\xi,\xi^{\opt};M)$,  is tightly bounded below by
\begin{eqnarray*}\label{eqn:saeff}
\eff_{\Linear}(\xi,\xi^{\opt};M) &=&  \frac{\tr\left[\mI^{-1}(\xi^{\opt};M)\mL\right]}{\tr\left[\mI^{-1}(\xi;M)\mL\right]}\nonumber\\
&=& \frac{\tr\left[\mI^{-1}(\xi^{\opt};M)\mL\right]}{\tr\left[\big(\mI({\tar};M)^{-1}\mI(\xi;M)\big)^{-1}\mI({\tar};M)^{-1}\mL\right]}\nonumber\\
&\geq & \frac{\tr\left[\mI^{-1}(\xi^{\opt};M)\mL\right]}{\rho\left[\left(\mI({\tar};M)^{-1}\mI(\xi;M)\right)^{-1}\right]\tr\left[\mI({\tar};M)^{-1}\mL\right]}\nonumber\\
&=& \frac{\eff_\Linear(\tar,\xi^{\opt};M)}{\rho\left[\left(\mI({\tar};M)^{-1}\mI(\xi;M)\right)^{-1}\right]}\nonumber\\
&\geq & (1-D(\xi;F_{\tar})\cdot V_{M})\eff_{\Linear}(\tar,\xi^{\opt};M),
\end{eqnarray*}
provided $D(\xi;F_{\tar})\cdot V_{M}<1$. When $D(\xi;F_{\tar})\cdot V_{M}\geq 1$, $(1-D(\xi;F_{\tar})\cdot V_{M})\leq 0$, and the inequality holds naturally.

\end{proof}
The tight lower bound \eqref{lowerbound} holds for all model specifications $M\in \mathcal{M}$, where the model space $\mathcal{M}$ is defined by the RKHS under the reproducing kernel $K$.
We will discuss the choice of the kernel function $K$ in the next section.
This tight lower bound implies that the L-efficiency of a design is always regulated by a reasonably large value for all $M\in\mathcal{M}$,
provided that the target distribution $F_{\tar}$ is chosen appropriately and the discrepancy $D(\xi,F_{\tar})$ is small enough.
It is worth pointing out that
the equality in the lower bound \eqref{lowerbound} holds for some model specification $M$ in the model space $\mathcal{M}$ defined by the RKHS.

It is seen that the lower bound of the L-efficiency consists of three terms,  $\eff_{\Linear}(\tar,\xi^{\opt};M)$, $V_{M}$, and $D(\xi;F_{\tar})$.
Both $\eff_{\Linear}(\tar,\xi^{\opt};M)$ and $V_M$ depend on link function $h$, basis functions $\vg$, regression coefficients $\vbeta$,  and the choice of target distribution $F_{\tar}$, but do not depend on the design of the pilot experiments.
The $\eff_{\Linear}(\tar,\xi^{\opt};M)$ is the L-efficiency of the continuous design with target distribution $F_{\tar}$, which depends on the choice of target distribution $F_{\tar}$ and model specification $M$.
The $V_M$,  as defined in Theorem \ref{thm:robust}, measures the roughness/degrees of oscillation of $\vg$ and $w$ of the model specification $M\in\mathcal{M}$.
Note that $V_{M}=\sup_{\|\valpha\|_2\leq 1} V(u_{\valpha;M})$ is the maximum semi-norm of $u_{\valpha;M}$, where  $u_{\valpha;M}: \vx\mapsto \valpha^\T[\mI({\tar};M)]^{-\frac{1}{2}}\mI_{\vx}[\mI({\tar};M)]^{-\frac{1}{2}}\valpha$ with $\mI_{\vx} = \vg(\vx)w(\vx;M)\vg^\top(\vx)$ and $\|\valpha\|_2\leq 1$.
Thus, $u_{\valpha;M}$
is a linear combination of $g_i(\vx)g_j(\vx)w(\vx;M)$, $i,j=1,\ldots,l$, and consequently $V_M$ can measure the roughness/degrees of oscillation of $g_i(\vx)g_j(\vx)w(\vx;M)$, $i,j=1,\ldots,l$.

The discrepancy $D(\xi;F_{\tar})$ depends on the design, the choice of target distribution $F_{\tar}$ and kernel $K$, but not on the model specification $M = (h,\vg,\vbeta)$.
It is to measure how well the design $\xi$ approximates the target distribution $F_{\tar}$.
Note that, only when $D(\xi;F_{\tar})\cdot V_{M}<1$, the lower bound $(1-D(\xi;F_{\tar})\cdot V_{M})\eff_{\Linear}(\tar,\xi^{\opt};M)$ of the L-efficiency makes sense in a practical perspective. That is, the L-efficiency is regulated by a reasonably large lower bound.
This condition implies that a design with a smaller discrepancy $D(\xi;F_{\tar})$ is required to regulate the L-efficiency if $\vg$ and $w$ of some model specifications $M\in\mathcal{M}$ are believed to be more oscillating.

Theorem \ref{thm:robust} provides a theoretical rationale to adopt a design that has a small discrepancy with respect to an appropriately chosen target distribution when little knowledge of the model specification is available to the experimenter.
The target distribution $F_{\tar}$ should be chosen so that the continuous design following $F_{\tar}$ obtains a reasonably large L-efficiency $\eff_{\Linear}(\tar,\xi^{\opt};M)$ for a variety of possible and relevant model specifications $M=(h,\vg,\vbeta)$. 
The choice of target distribution $F_{\tar}$ for the pilot experiments of GLMs will be discussed in the next section.

In practice, the experimenter could also be interested in regulating the prediction error of the response. 
Under this consideration, a prediction-oriented optimality criterion for GLMs, such as the EI-optimality in \citep{li2020efficient}, can be adopted to investigate designs for the pilot experiments.
The EI-optimality, as a flexible generalization of classical I-optimality, aims at minimizing the integrated mean squared prediction error of the response with respect to some measure $F_{\IMSE}$.
Specifically, the EI-optimality in \citep{li2020efficient} is expressed as
$$\EI(\xi;M,F_{\IMSE}) = \tr\left(\mA_{(M,F_{\IMSE})} \mI(\xi;M)^{-1}\right),$$
with matrix $\mA_{(M,F_{\IMSE})} =\int_{\Omega}\vg(\vx)\vg^T(\vx)\left[\frac{\dif h^{-1}}{\dif \eta}\right]^2\dif F_{\IMSE}(\vx)$ depending only on the model specification $M$ and $F_{\IMSE}$, but not the design $\xi$.
The following result shows that the low discrepancy design can also regulate the EI-optimality over a variety of model specifications.
It implies that a design with a small discrepancy is also desirable when the objective of the pilot experiments is to control the prediction error.

\begin{cor}\label{cor:robust}
In the same conditions of Theorem \ref{thm:robust}, for a design $\xi$, the EI-efficiency for any model specification $M\in\mathcal{M}$ is tightly bounded below by
$$\eff_{\EI}(\xi,\xi^{\opt};M, F_{\IMSE})\geq (1-D(\xi;F_{\tar})\cdot V_{M})\eff_{\EI}(\tar,\xi^{\opt};M, F_{\IMSE}),$$
where $\eff_{\EI}(\xi,\xi^{\opt};M, F_{\IMSE}) = \frac{\EI(\xi^{\opt};M,F_{\IMSE})}{\EI(\xi;M,F_{\IMSE})}$ and $\eff_{\EI}(\tar,\xi^{\opt};M, F_{\IMSE}) = \frac{\EI(\xi^{\opt};M,F_{\IMSE})}{\EI(\tar;M,F_{\IMSE})}$ are the EI-efficiency of design $\xi$ and that of a continuous design following target distribution $F_{tar}$ relative to locally optimal design $\xi^{\opt}$ of model specification $M$, respectively.
\end{cor}

\begin{proof}
The proof is the same as the proof of Theorem \ref{thm:robust}, where $\mL = \mA_{(M,F_{\IMSE})}$ in this corollary.
\end{proof}

With a specific target distribution $F_\tar$ and reproducing kernel $K$,
the design $\xi$ is called \emph{a low discrepancy design} if its corresponding discrepancy $D(\xi; F_{\tar})$ is small.
Low discrepancy designs are prevalent in the Monte Carlo community to estimate high-dimensional integration \citep{novak2001integration}.
For the uniform target distribution $F_{\tar} = F_{\unif}$ on $[0,1]^d$,
Sobol sequences are usually preferred since they can be constructed easily and asymptotically achieve a small discrepancy $D(\xi;F_{\unif})$ under popular choices of kernel function $K$.
Regarding the construction of non-uniform low-discrepancy design, i.e.  non-uniform $F_{\tar}$,
the common practice is to use the inverse transformation of a uniform low-discrepancy design on $[0,1]^d$.
However, \cite{li2020transformed} showed that considering the uniform and non-uniform discrepancies defined by the same reproducing kernel $K$,
the inverse transformed uniform low discrepancy design may not preserve a small discrepancy for the non-uniform target distribution.
Other than the existing low discrepancy designs such as Sobol sequences and their inverse transformations, one could also construct a design that minimizes $D(\xi;F_{\tar})$ using optimization methods.
\cite{winker1997application} and \cite{fang2000uniform} proposed a threshold acceptance algorithm, and \cite{li2020transformed} developed a coordinate-exchange algorithm to construct such designs, but that is beyond the scope of this paper.

\subsection{Choice of Target Distribution and Reproducing Kernel}\label{sec:comparediscrepancy}
In the lower bound \eqref{lowerbound},  the L-efficiency is regulated for all model specifications in the model space $\mathcal{M}$, which is determined by the reproducing kernel $K$.  In a pilot experiment, one should choose the reproducing kernel $K$ whose corresponding RKHS contains most of the commonly used model specifications for GLMs. Furthermore,  to achieve a large lower bound in \eqref{lowerbound}, one should choose a target distribution $F_{\tar}$ that achieves a reasonably large L-efficiency $\eff_{\Linear}(\tar,\xi^{\opt};M)$ for the model specifications under consideration.  After deciding the reproducing kernel $K$ and target distribution $F_{\tar}$,  a design with a small discrepancy $D(\xi;F_{\tar})$ should be adopted in the pilot experiments for GLMs.  Note that the discrepancy in \eqref{discdef} depends on both $K$ and $F_{\tar}$. In this section, we would discuss the choice of target distribution $F_{\tar}$ and the reproducing kernel $K$.

The target distribution $F_{\tar}$ should be chosen as the one with a reasonably large L-efficiency $\eff_{\Linear}(\xi^{\tar},\xi^{\opt};M)$  for a variety of possible model specifications $M=(h,\vg,\vbeta)$ that one believes to be most relevant. Without loss of generality, the experimental region is assumed to be $\Omega=[-1,1]^d$, and two target distributions are considered in this work.
One is the uniform distribution $F_{\unif}$. For GLMs, the basis function $\vg$ may include interactions or higher-order polynomials, and the corresponding locally optimal design points are usually quite evenly located in the experimental region.
What's more, the uniform low discrepancy designs, such as Sobol sequence, are already available and ready to use \citep{niederreiter1988low, owen2000monte}.
The other target distribution considered is the arcsine distribution $F_{\asin}$ with density function $f_{\asin}(\vx) = \prod_{i=1}^d\frac{1}{\pi\sqrt{1-x_d^2}},  \vx\in[-1,1]^d$.  Unlike uniform distribution, the arcsine distribution tends to push the points towards the edges of the experimental region. For univariate linear regression models, it has been shown that arcsine support designs achieve large A-efficiency for polynomial basis functions \citep{93puk}.  We will compare the performance of the uniform low-discrepancy  and arcsine low-discrepancy designs in Section \ref{sec:examples}.

Regarding the choice of reproducing kernel $K$, we consider a popular reproducing kernel \citep{hickernell1998generalized} for experimental region $[-1,1]^d$
\begin{equation}\label{eqn:kernelchoice}
K(\vx,\vz) = \prod\limits_{j=1}^d \left[1+\frac{1}{2}(|x_j|+|z_j|-|x_j-z_j|)\right].
\end{equation}
The corresponding RKHS induced by this kernel contains all functions whose mixed partial derivatives up to order one in each coordinate are square-integrable.
Such an RKHS will include most of the commonly used model specifications in practice, such as logit, probit link functions with main-effect, interactions, and higher-order polynomial basis functions and arbitrary finite regression coefficient values.
Another advantage of this kernel is that the corresponding discrepancy is invariant under reflections of the design about any plane $x_j=0$ for a symmetric target distribution.  Thus, we choose to use kernel \eqref{eqn:kernelchoice} in this work.

Then, the corresponding discrepancy of a design $\xi$ on $[-1,1]^d$ with respect to $F_{\unif}$ and $F_{\asin}$ can be expressed as,
\begin{align}\label{eqn:unifdiscrepancy}
D^2(\xi; F_\unif)
& = \left(\frac{7}{6}\right)^d - \frac{1}{2^{d-1}n}\sum_{i=1}^m n_i \prod_{j=1}^d \left[2+|x_{ij}|-\frac{x_{ij}^2}{2} \right]\nonumber \\
& \qquad + \frac{1}{n^2}\sum_{i,k=1}^m n_in_k\prod_{j=1}^d\left[1+\frac{1}{2}\left(|x_{ij}|+|x_{kj}|-|x_{ij}-x_{kj}| \right) \right],
\end{align}
and
\begin{align}\label{eqn:asindiscrepancy}
D^2(\xi; F_\asin)
& = \left(1+\frac{2}{\pi}-\frac{4}{\pi^2}\right)^d - \frac{2}{n}\sum_{i=1}^m n_i \prod_{j=1}^d \left[1+\frac{1}{\pi}+\frac{1}{2}|x_{ij}|-\frac{1}{\pi}\left(x_{ij}\arcsin(x_{ij})+\sqrt{1-x_{ij}^2} \right)\right]\nonumber \\
& \qquad + \frac{1}{n^2}\sum_{i,k=1}^m n_in_k\prod_{j=1}^d\left[1+\frac{1}{2}\left(|x_{ij}|+|x_{kj}|-|x_{ij}-x_{kj}| \right) \right],
\end{align}
respectively.
The derivation of the discrepancy in \eqref{eqn:unifdiscrepancy} and \eqref{eqn:asindiscrepancy} is provided in the appendix.

\section{Numerical Examples}\label{sec:examples}
We first consider the uniform distribution as the target distribution.
The discrepancy of a design $\xi$, $D(\xi; F_{\unif})$, as a measure of the difference between the empirical distribution of $\xi$ and the uniform distribution $F_{\unif}$,  reflects the space-filling property of $\xi$ to some extent.
Besides the Sobol sequence, there are other popular space-filling designs in the literature.
For example, the Latin hypercube design is a space-filling design with one-dimensional stratification property \citep{mckay2000comparison}.
There are two popular types of Latin hypercube designs.
One is the maximin Latin hypercube design proposed by \cite{morris1995exploratory}, which maximizes the minimum Euclidean distance between any two points in the design.
The other type is the design that minimizes the correlations among experimental factors \citep{iman1982distribution,owen1994controlling,tang1998selecting}.
\cite{joseph2015maximum} proposed a maximum projection design that optimizes projection properties on all subspaces of experimental factors.
In this work, we consider five space-filling designs for comparison: (1) scrambled Sobol design (SSD); (2) maximin Latin hypercube design (MmLHD); (3) correlation minimized Latin hypercube design (mcLHD); (4) maximum projection Latin hypercube design (MPLHD); and (5) random design (Random).
A scrambled Sobol design is a randomly scrambled Sobol sequence discovered by \cite{owen2000monte}, which achieves better equidistribution of nets  compared to a deterministic Sobol sequence \citep{hickernell1996mean}.
The scrambled Sobol,  Latin hypercube and random designs are generated using existing \textsc{Matlab} routines, \emph{sobolset},  \emph{lhsdesign} and \emph{rand}, respectively.
Maximum projection Latin hypercube designs are generated using R package \emph{MaxPro} (v4.1-2; Shan and Joseph, 2018).
Note that the classical Sobol sequence and Latin hypercube designs are usually on $[0,1]^d$,
and here the above five designs are generated using the classical space-filling designs with appropriate scale and shift to match $\Omega = [-1,1]^d$.
In addition to space-filling designs, we also consider the designs to approximate arcsine target distribution $F_{\asin}$.
To construct such designs,  we use the inverse transformed space-filling designs on $[0,1]^d$, i.e., (6) AsinSSD; (7) AsinMmLHD; (8) AsinmcLHD; (9) AsinMPLHD; and (10) AsinRandom. As stated in Section \ref{sec:boundleff},
one can also construct the designs that minimize $D(\xi;F_{\unif})$ or $D(\xi;F_{\asin})$. For the sake of computation efficiency, we choose to use the readily available designs in this work.

In this section, we conduct several numerical examples to examine the discrepancy (under either $F_{\unif}$ or $F_{\asin}$) and the L-efficiency of the above ten designs.
Since the accuracy of the coefficient estimates is usually of interest in pilot experiments,
we choose $\mL$  in L-optimality to be the identity matrix, and it becomes the popular A-optimality.
Note that the number of replications does not affect the L-efficiency of a design,
In the following examples, it is assumed that the pilot experiment consists of $n$ distinct points.
Since all ten types of designs in comparison are random, we generate 100 sets of design points for each of them and compute the average A-efficiency.

\textbf{Example 1.} We consider the crystallography experiment example in \cite{woods2006designs}, which studies how process variables affect the probability that a new product is formed in a crystallography experiment.
The four explanatory variables $-1\leq x_i\leq 1$, $i=1,...,4$ are rate of agitation during maxing, volume of composition, temperature, and evaporation rate, and the binary response $Y\in \{0,1\}$ denotes whether a new product is formed. A logistic regression model
$$Prob(Y=1|\vx) = \frac{e^{\eta(\vx)}}{1+e^{\eta(\vx)}}.$$
with main-effect only, $\eta(\vx) = \beta_0+\sum\limits_{i=1}^4 \beta_ix_i$, is used.
We consider a pilot experiment with $n=2^4 = 16$ distinct design points.
We investigate the performance of the ten designs over three coefficient spaces $\mathcal{B}_1$, $\mathcal{B}_2$, and $\mathcal{B}_3$, the details of which are provided in Table \ref{tab:parameterspace}. Here, $\mathcal{B}_2$ has the same centroid as $\mathcal{B}_1$ but substantially smaller volume, and $\mathcal{B}_3$ has the same volume as $\mathcal{B}_1$ but is centered further from $(0,\cdots,0)$.  These are the model and coefficient spaces used in \cite{woods2006designs}.
To assess the performance of the ten designs, we calculate the A-efficiency of each design over a grid sample of size $N = 7^5 = 16,807$ drawn from each coefficient space.

\begin{table}[htbp]
  \centering
  \caption{Ranges of Regression Coefficients for the Coefficient Space $\mathcal{B}_j$, $j=1,2,3$}
    \begin{tabular}{lccc}
    \hline
     Regression      &   \multicolumn{3}{c}{Coefficient space}  \\
          \cline{2-4}
    coefficient  & $\mathcal{B}_1$  & $\mathcal{B}_2$  & $\mathcal{B}_3$ \\
    \hline
    $\beta_0$  &  $[-3,3]$     &  $[-1,1]$     & $[-3,3]$  \\
    $\beta_1$  &  $[-2,4] $    &  $[0,2]$     &  $[4,10]$ \\
    $\beta_2$   & $[-3,3]$      & $[-1,1]$      &  $[5,11]$ \\
    $\beta_3$   & $[0,6]$      &  $[2,4]$     &  $[-6,0]$ \\
    $\beta_4$   &  $[-2.5, 3.5]$     &    $[-.5,1.5]$   & $[-2.5,3.5]$ \\
    \hline
    \end{tabular}%
  \label{tab:parameterspace}%
\end{table}
The boxplots of A-efficiency and discrepancy of the ten designs are reported in Figure \ref{fig:cryAopt} and \ref{fig:cryDisc}, respectively.  The red asterisks (*) in Figure \ref{fig:cryAopt} represent the worst-case A-efficiency of each design over $N = 16,807$ sampled regression coefficients. In general, the designs with smaller discrepancy, either with respect to uniform distribution or arcsine distribution, yield larger A-efficiency for all three coefficient spaces.    Specifically, it is seen from Figure \ref{fig:cryDisc} that
the scrambled Sobol design (SSD) and the maximum projection Latin hypercube design (MPLHD) obtain smaller uniform discrepancies than the other space-filling designs, and similarly, their arcsine counterparts obtain smaller arcsine discrepancies.  A consistent pattern is observed in Figure \ref{fig:cryAopt} regarding A-efficiency. Compared to the other designs of the same target distribution, SSD, MPLHD, AsinSSD, and AsinMPLHD also achieve larger worst-case A-efficiency.  These numerical results confirm our theoretical understanding that the design with a smaller discrepancy performs better.  Furthermore,  Figure \ref{fig:cryAopt} shows that the uniform low-discrepancy designs, which adopt the uniform distribution $F_{\unif}$ as the target distribution, are superior in the perspective of the worst-case A-efficiency.  That is,  to regulate the worst-case A-efficiency,  the uniform distribution is a more preferable choice for target distribution $F_{\tar}$,  although under some circumstances ($\mathcal{B}_1$ and $\mathcal{B}_2$), choosing arcsine distribution $F_{\asin}$ as the target distribution yields larger the median and maximum A-efficiency.


\begin{figure}[htbp]
\centering
\subfloat[$\mathcal{B}_1$]
{{\includegraphics[width=18cm,height=6.5cm]{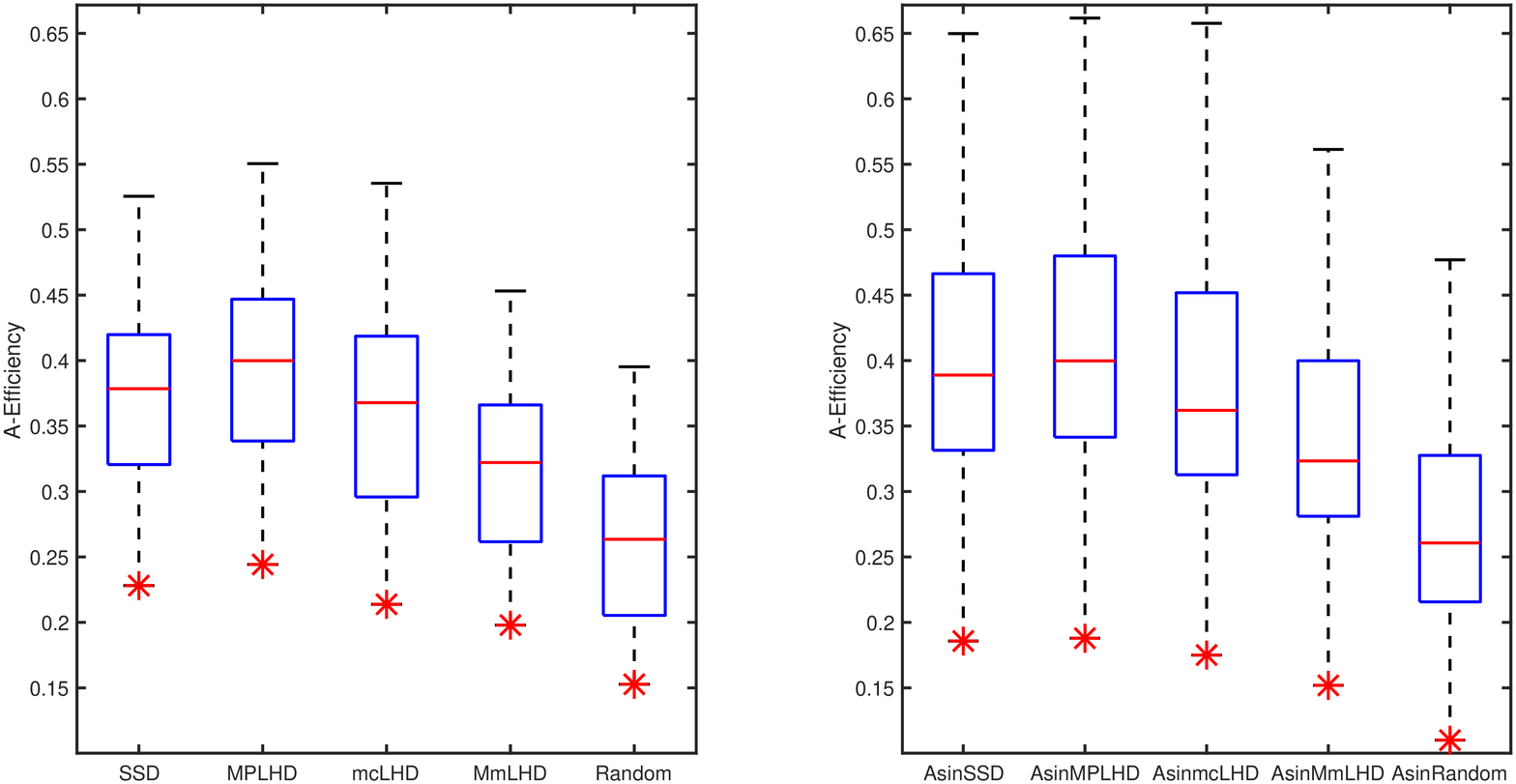}}}
\qquad
\subfloat[$\mathcal{B}_2$]
{{\includegraphics[width=18cm,height=6.5cm]{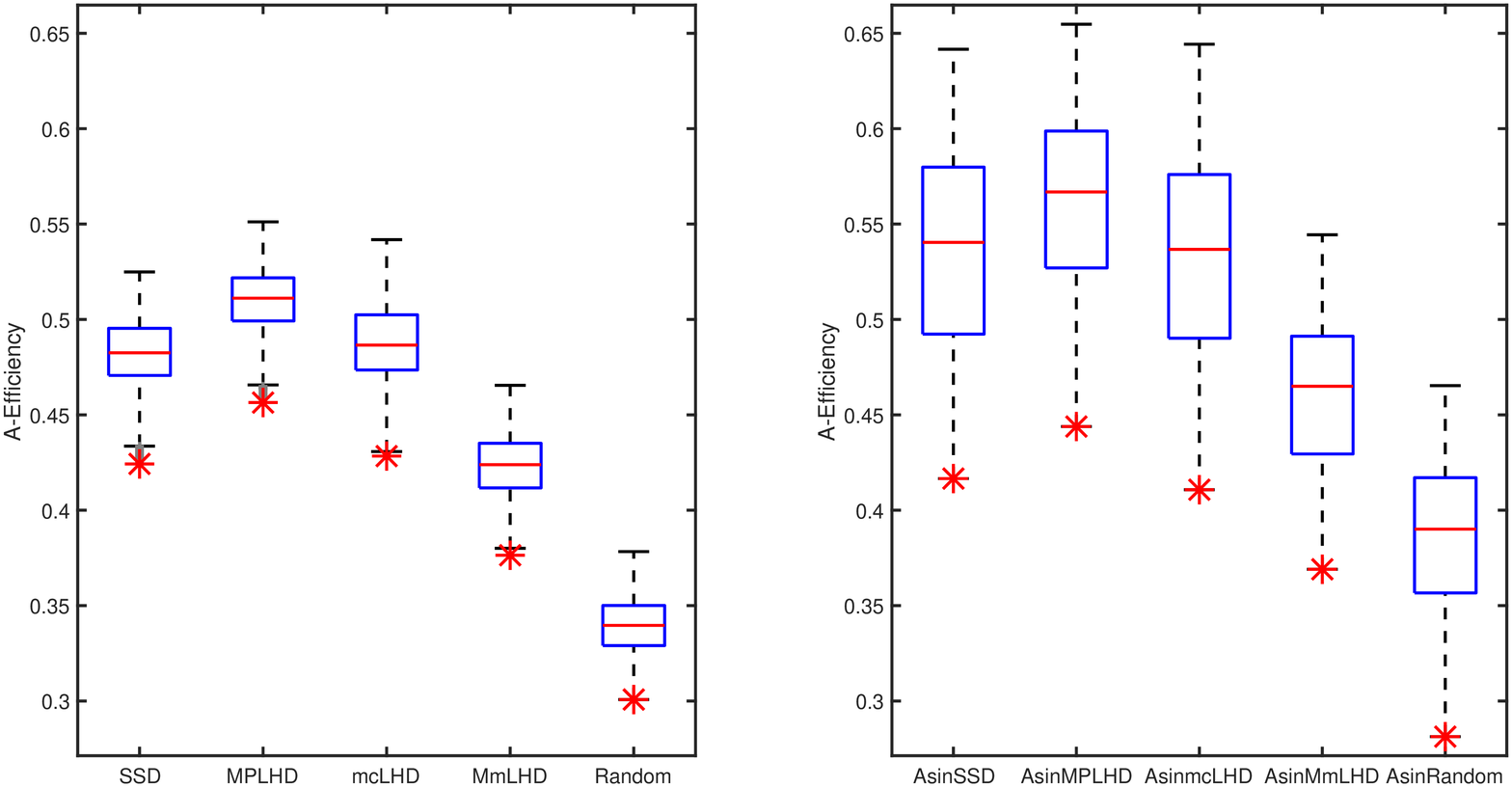}}}
\qquad
\subfloat[$\mathcal{B}_3$]
{{\includegraphics[width=18cm,height=6.5cm]{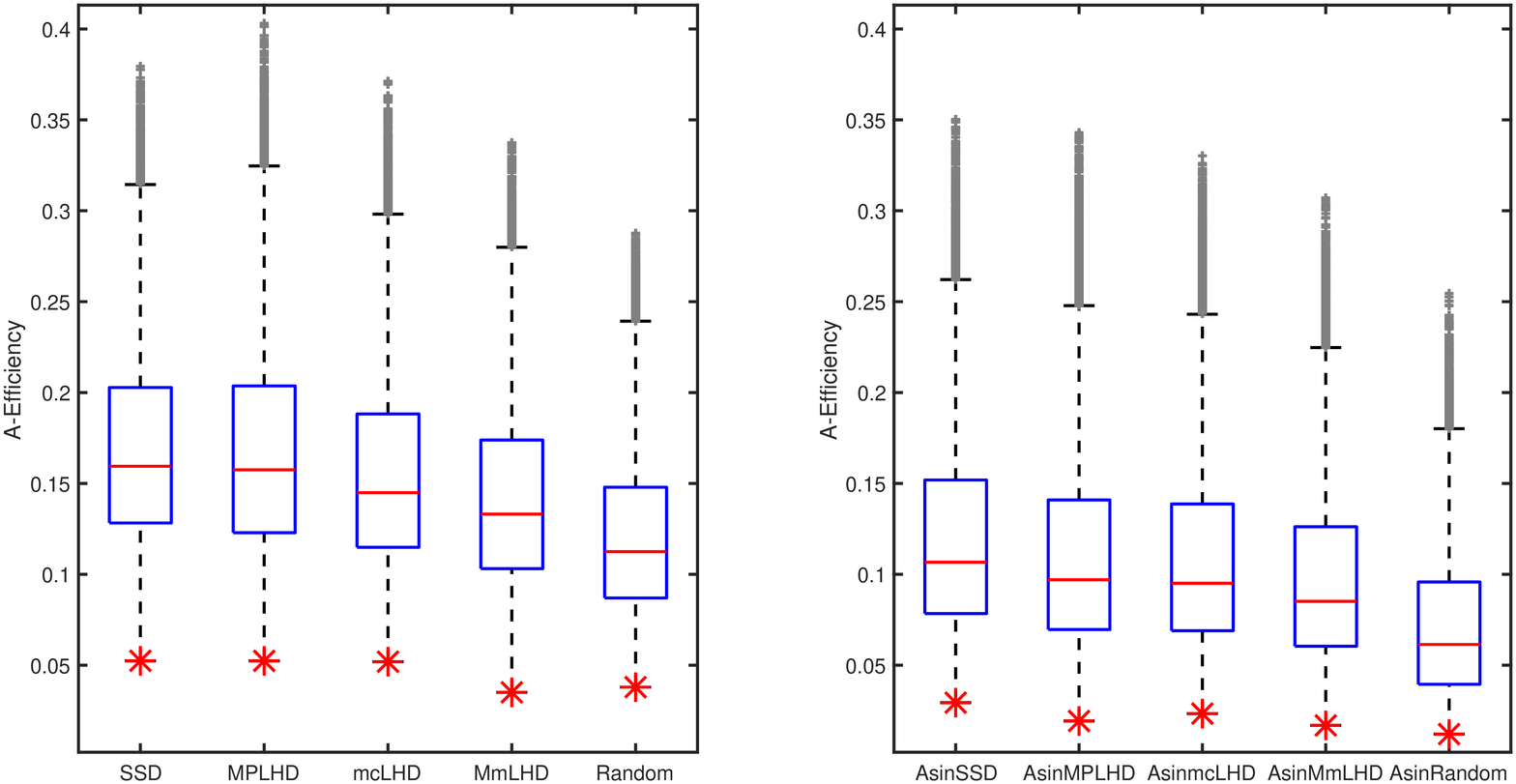}}}
\caption{Boxplot of A-efficiency of ten designs over $N=16,807$ regression coefficients}
\label{fig:cryAopt}
\end{figure}

\begin{figure}[htbp]
\centering
{{\includegraphics[width=18cm,height=6.5cm]{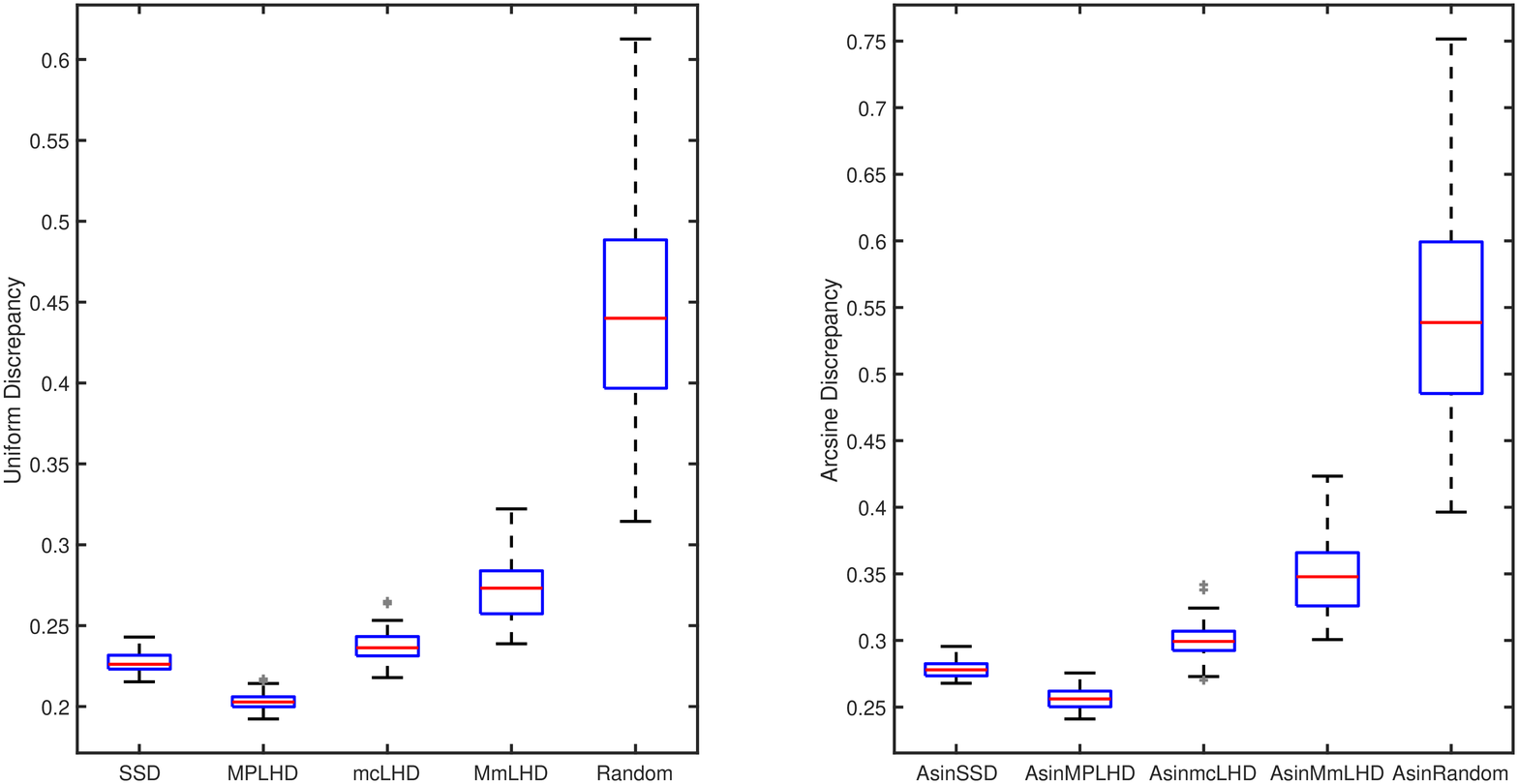}}}
\caption{Boxplot of uniform discrepancy of space-filling designs and arcsine discrepancy of arcsine transformed space-filling designs for $d=4$ and $n=16$}
\label{fig:cryDisc}
\end{figure}



\textbf{Example 2.} In this example, we investigate the performance of the space-filling designs and the arcsine transformed designs with $n=32$ distinct design points for the probit model with $d=6$ experimental factors $\vx = [x_1,\ldots,x_6]$ on $\Omega = [-1,1]^6$. Such probit model with a binary response $Y\in\{0,1\}$ can be expressed as
$$Prob(Y=1|\vx) = \Phi(\eta(\vx)),$$
where $\Phi$ is the cumulative distribution function of standard normal distribution.
Two linear predictors, one with only main effects, and the other with some interactions, are considered:
\begin{align*}
\text{Linear Predictor 1}: & \eta_1(\vx) = \beta_0+\sum\limits_{i=1}^6 \beta_ix_i, \\
\text{Linear Predictor 2}: & \eta_2(\vx)= \beta_0+\sum\limits_{i=1}^6 \beta_ix_i+\theta_1x_1x_2+\theta_2x_2x_3+\theta_3x_4x_6.
\end{align*}

To assess the performance of the ten designs, the range of each regression coefficient is set to be $[-1.2,1.2]$,
from which a Sobol sample of $N=1,024$ values of regression coefficients is generated.
The A-efficiency of the ten designs over $N=1,024$ samples of regression coefficients for linear predictors 1 and 2, and their discrepancies are computed.
The corresponding boxplots of A-efficiency and design discrepancy are reported in Figure \ref{fig:probitAeff} and \ref{fig:discrepancy_d6n32}, respectively.
\begin{figure}[hbtp]
\centering
\subfloat[Linear Predictor 1]
{{\includegraphics[width=18cm,height=6.5cm]{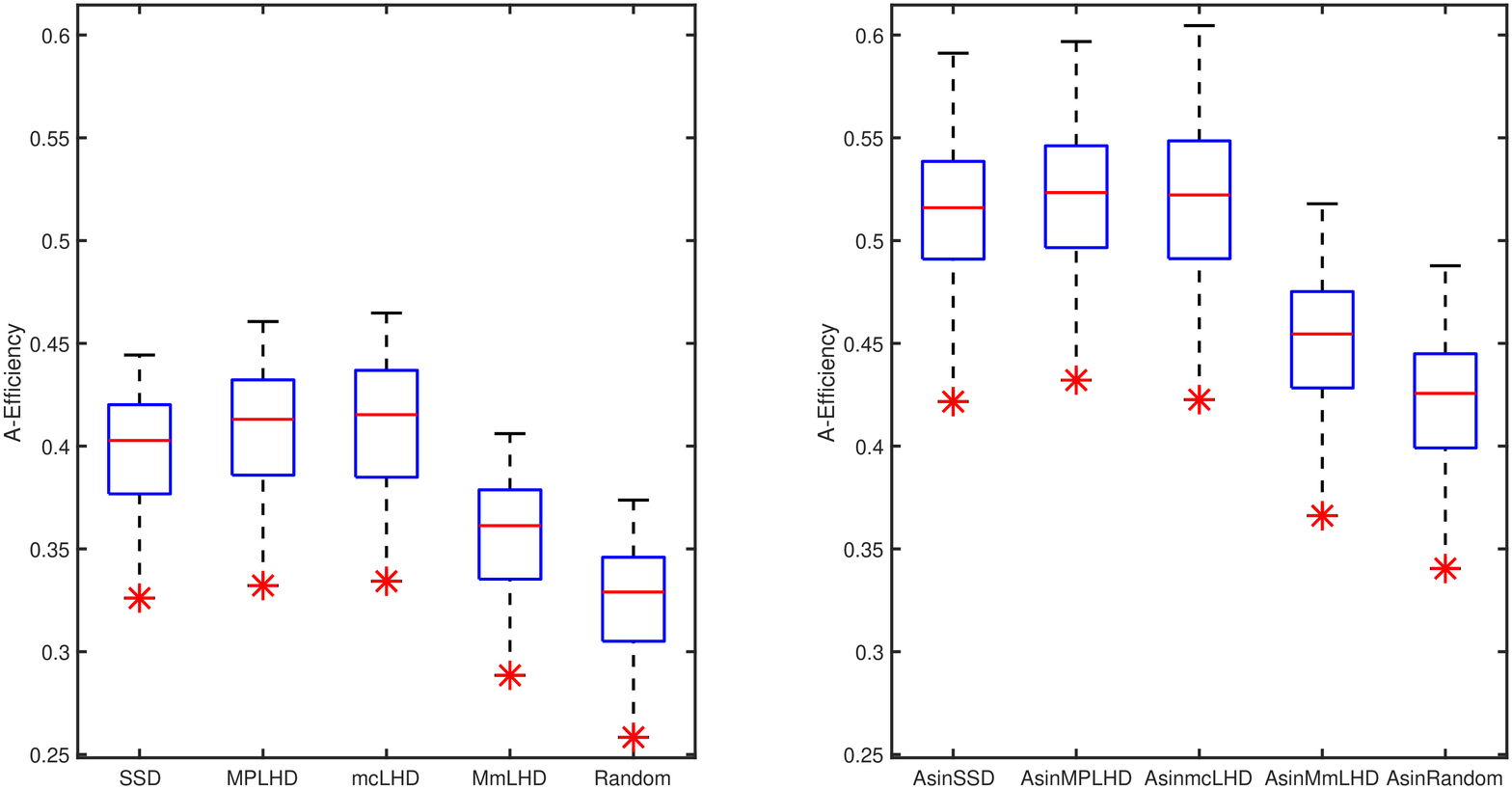}}}
\qquad
\subfloat[Linear Predictor 2]
{{\includegraphics[width=18cm,height=6.5cm]{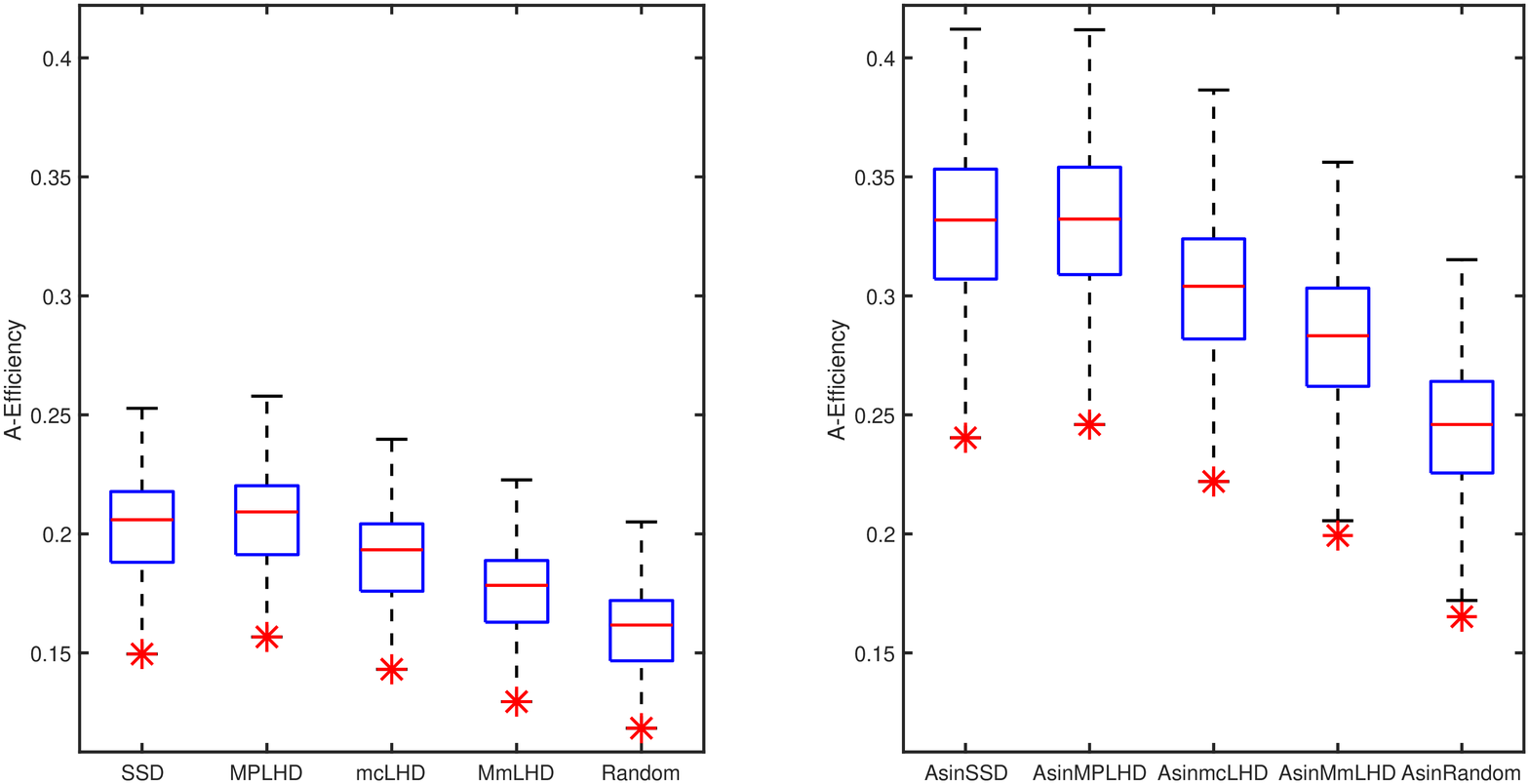}}}
%
\caption{Boxplot of A-efficiency of the ten designs over $N=1024$ regression coefficients}
\label{fig:probitAeff}
\end{figure}

\begin{figure}[hbtp]
\centering
{\includegraphics[width=18cm,height=6.5cm]{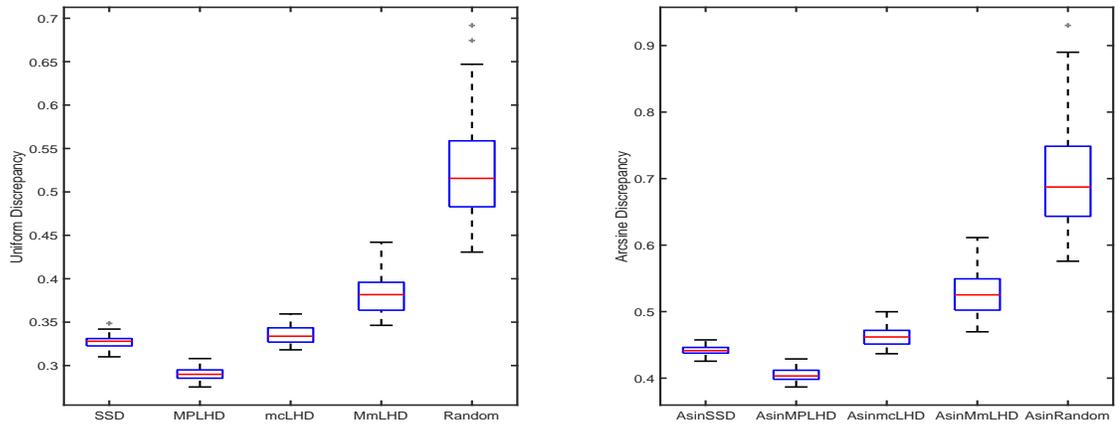}}
\caption{Boxplot of uniform discrepancy of space-filling designs and arcsine discrepancy of arcsine transformed space-filling designs for $d=6$ and $n=32$}
\label{fig:discrepancy_d6n32}
\end{figure}

Figure \ref{fig:probitAeff} and \ref{fig:discrepancy_d6n32} reveal that the performance on the A-efficiency of a design is consistent with that on the discrepancy, which confirms the theoretical result that a design with a smaller discrepancy tends to have a larger worst-case A-efficiency. For the choice of the target distribution, interestingly,  different from the observations in Example 1 of logistic regression, the designs with small arcsine discrepancy obtain larger worst-case A-efficiency for both predictors 1 and 2.  From Figure \ref{fig:probitAeff}-(a), it is seen that AsinSSD, AsinMPLHD, and AsinmcLHD are comparable and better than the other designs in comparison.
When the interactions are involved in the linear predictor 2 of the probit model,
Figure \ref{fig:probitAeff}-(b) shows that AsinSSD and AsinMPLHD give the comparable performance and are better than the other designs in comparison.


\textbf{Example 3.} This example explores the performance of the uniform/arcsine low discrepancy designs under the linear regression model.
We consider the design variable $\vx$ on $[-1,1]^7$ of $d=7$, and the response of a linear regression model is
$$y = \vbeta^{\top}\vg(\vx) + \epsilon,$$
where the noise term $\epsilon\sim N(0,\sigma^2)$, and the noises corresponding to different experimental variable values are assumed to be independent.
Note that the A-efficiency of a design for linear regression model depends only on the basis function $\vg$,
but not the regression coefficient $\vbeta$.
We consider three types of the basis function $\vg$ as follows:
\begin{align*}
\textrm{Main effect only: } &[1, x_1,\ldots,x_7]^{\top}; \\
\textrm{Main effect with one second-order term: } &[1, x_1,\ldots,x_7, x_ix_j]^{\top}; \\
\textrm{Main effect with two interactions: } &[1,x_1,\ldots, x_7, x_ix_j, x_kx_s]^{\top},
\end{align*}
where $i,j,k,s = 1,\ldots, 7,$ and  $i\neq j\neq k\neq s$. There are totally $N = 174$ basis functions.
The A-efficiency of the ten designs of $n= 2^7 = 128$ distinct points over the considered basis functions are computed.
Figure \ref{fig:linearAeff} and \ref{fig:discrepancy_d7n128} report the  A-efficiency and discrepancy performance of the ten designs, respectively. Generally speaking, the performance of A-efficiency and that of the discrepancy is consistent, which again confirms our theoretical understanding that a design with low discrepancy is suitable when little model specification information is available.  From Figure \ref{fig:linearAeff}, it can be seen that the designs approximating arcsine target distribution perform better than the ones approximating uniform target distribution, which echos the known results for 1-$d$ linear regression models that arcsine support designs obtain high A-efficiency \citep{93puk}.  Among the ten designs,  the arcsine inverse transformed scrambled Sobol design (AsinSSD), which has the smallest arcsine discrepancy, provides the best performance regarding the worst-case A-efficiency.

We would like to point out that,  SSD and AsinSSD are very easy to construct and does not require complicated optimization,
while both MPLHD and mcLHD involve nonlinear optimization in the design construction.
As seen from Figure \ref{fig:cryDisc}, \ref{fig:discrepancy_d6n32} and \ref{fig:discrepancy_d7n128}, the Sobol sequence becomes more advantageous as the design size increases.

\begin{figure}[hbtp]
\centering
{\includegraphics[width=18cm,height=6.5cm]{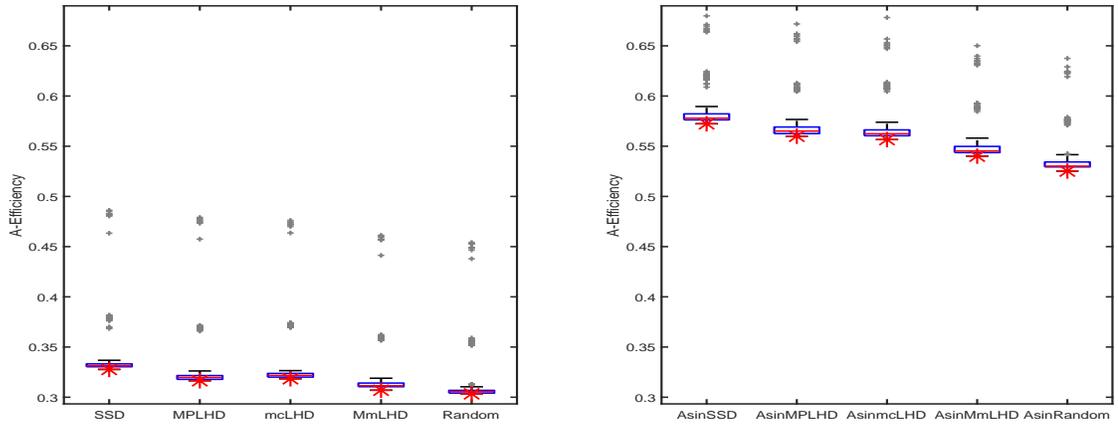}}
\caption{Boxplot of  A-efficiency of the ten designs over $N=174$ basis functions for $d=7$ and $n=128$}
\label{fig:linearAeff}
\end{figure}
\begin{figure}[hbtp]
\centering
{\includegraphics[width=18cm,height=6.5cm]{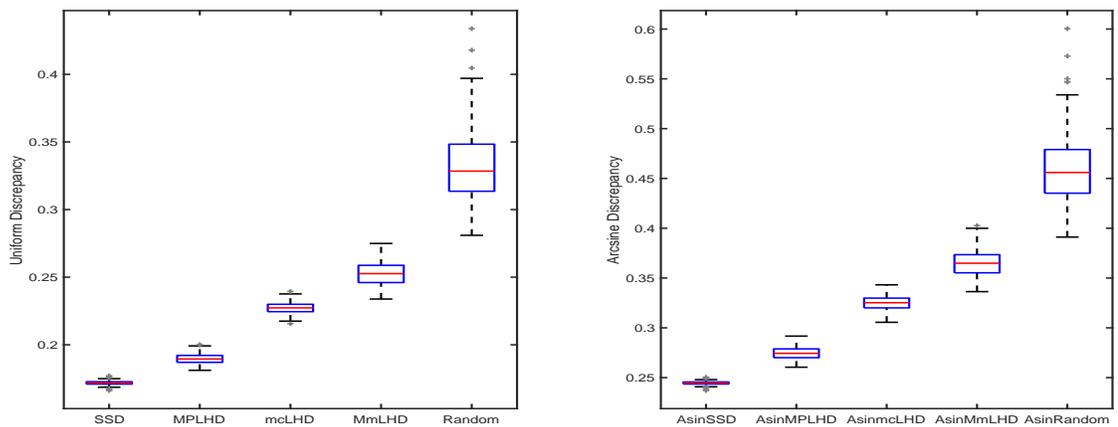}}
\caption{Boxplot of uniform discrepancy of space-filling designs and arcsine discrepancy of arcsine transformed space-filling designs for $d=7$ and $n=128$}
\label{fig:discrepancy_d7n128}
\end{figure}

\section{Discussion}\label{sec:conclusion}
In this work, we investigate the designs for the pilot experiments of GLMs, when little information of the model specification is available.
By deriving a tight lower bound on the L-efficiency of a design for all model specifications in a model space, it is seen that a low discrepancy design would regulate the L-efficiency for a variety of model specifications, and thus is well suited to the purpose of the pilot experiments of GLMs.
Through the numerical comparison of space-filling designs and their arcsine inverse transformed designs, it is observed that a design with a small discrepancy regulates the worst-case A-efficiency.  Among the compared designs,  the maximum projection Latin hypercube design \citep{joseph2015maximum}, which requires some heuristic optimization procedure,  possesses a low uniform discrepancy when the design size is small.  Similarly, the arcsine inverse transformation of a small-sized maximum projection Latin hypercube design usually has a low arcsine discrepancy.
While the scrambled Sobol design, which is easy and fast to construct,  asymptotically achieves low uniform discrepancy. Similarly, its arcsine counterpart usually preserves a small arcsine discrepancy asymptotically.
In the numerical study, it is observed that a design with a small uniform discrepancy obtains a larger worst-case A-efficiency for logistic regression, while an arcsine low-discrepancy design is superior for probit and linear regression models.

Through the pilot experiment, one can obtain some informative understanding of the model specifications,
which can facilitate the next phase of experimental design over a much smaller model space $\mathcal{M}$.
If a single $M = (h,\vg,\vbeta)$ is of interest, one can construct the locally optimal design that optimizes some design criterion, or equivalently, maximizes the corresponding design efficiency.
If multiple model specifications are considered, one can adopt the Bayesian designs \citep{atkinson2015designs} and compromise designs \citep{woods2006designs},
which are globally optimal designs that maximize the mean design efficiency over all potential model specifications with some pre-assumed prior distribution of the potential model specifications.
Alternatively, one can consider the maximin designs \citep{imhof2000graphical, li2020maximin}, which aim at maximizing the minimum design efficiency over all potential model specifications.

There are several directions for future research of the design for pilot experiments.
First, many experiments in engineering and health care encounter both quantitative and qualitative (QQ) responses \citep{deng2015qq, kang2018bayesian}.
It would be interesting to investigate the efficient designs for pilot experiments with QQ response.
Second,  the designs of mixture experiments \citep{shen2020additive} have a constraint design space since the proportions of blends in mixture experiments must sum up to one.
It is not clear how to construct low-discrepancy designs on the constrained space
and whether the constructed designs are still efficient for the pilot study of mixture experiments.
Third, the pilot experiments are also needed in the application of the recommender system, which is of great importance in e-commerce to make customized recommendations for users.
The proposed low-discrepancy design can also be used in such applications, with the challenges that the underlying model for matrix completion contains both linear model and low-rank matrix \citep{mao2019matrix, zeng2021design}.
Finally, it would be interesting to study the connection between the robust design in \cite{dean2015handbook} and the low-discrepancy design in this work.
\cite{hickernell2002uniform} showed that both the variance and the bias of the response prediction are regulated using a low-discrepancy design for linear regression models with misspecification.
A future direction could be to investigate how the low-discrepancy design can regulate bias and variance of the response prediction for generalized linear models when misspecification is considered.

\section*{Acknowledgements}
The  authors  would  like  to  sincerely  thank  the  Associate  Editor  and  reviewers  for  their  insightful  comments.  
Deng's work was partly supported by National Science Foundation CISE Expedition grant CCF-1918770.

\section*{Appendix}
\emph{Derivation of the discrepancy in \eqref{eqn:unifdiscrepancy}.}

\noindent We first consider the case $d=1$. We integrate the kernel once:
\begin{align*}
\int_{-1}^1 K(t,x) \, \dif F_{\unif}(t)=&
 \frac{1}{2}\int_{-1}^1 \left[1+\frac{1}{2}(|t|+|x|-|t-x|)\right] \, \dif t\\
=& \frac{1}{2}\left[2+|x|+\frac{1}{2}-\frac{1}{2}\left(\int_{-1}^x(x-t)\dif t+\int_x^1 (t-x)\dif t\right)\right]\\
=&\frac{1}{2}\left[\frac{5}{2}+|x|-\frac{1}{2}\left(x^2+1\right)\right]\\
=& \frac{1}{2}\left(2+|x|-\frac{1}{2}x^2\right).
\end{align*}
Then we integrate once more:
\begin{align*}
{\int_{-1}^1 \int_{-1}^1 K(t,x) \, \dif F_{\unif}(t) \dif F_{\unif}(x)} &= \int_{-1}^{1}  \frac{1}{4}\left(2+|x|-\frac{1}{2}x^2\right) \, \dif x\\
&= \frac{7}{6}.
\end{align*}

Generalizing this to the $d$-dimensional case yields
\begin{gather*}
\int_{[-1,1]^d\times [-1,1]^d} K(\vx,\vt) \, \dif F_{\unif}(\vx)\dif F_{\unif}(\vt) = \left(\frac{7}{6}\right)^d, \\
\int_{[-1,1]^d}K(\vx,\vx_i) \, \dif F_{\unif}(\vx) = \frac{1}{2^d}\prod\limits_{j=1}^d \left(2+|x_{ij}|-\frac{1}{2}x_{ij}^2\right).
\end{gather*}
Thus, the discrepancy of a design $\xi$ for the uniform distribution on $[-1,1]^d$ is
\begin{align*}
D^2(\xi; F_\unif)
& = \left(\frac{7}{6}\right)^d - \frac{1}{2^{d-1}n}\sum_{i=1}^m n_i \prod_{j=1}^d \left[2+|x_{ij}|-\frac{x_{ij}^2}{2} \right]\nonumber \\
& \qquad + \frac{1}{n^2}\sum_{i,k=1}^m n_in_k\prod_{j=1}^d\left[1+\frac{1}{2}\left(|x_{ij}|+|x_{kj}|-|x_{ij}-x_{kj}| \right) \right],
\end{align*}

\noindent \emph{Derivation of the discrepancy in \eqref{eqn:asindiscrepancy}.}

Following the same procedure as the derivation of $D^2(\xi; F_\unif)$,
$$\int_{-1}^1 K(t,x) \, \dif F_{\asin}(t) = 1+\frac{1}{\pi}+\frac{1}{2}|x|-\frac{1}{\pi}(x\arcsin(x)+\sqrt{1-x^2}),$$
$${\int_{-1}^1 \int_{-1}^1 K(t,x) \, \dif F_{\asin}(t) \dif F_{\asin}(x)} = 1+\frac{2}{\pi}-\frac{4}{\pi^2},$$
and thus, Thus, the discrepancy of a design $\xi$ for the arcsine distribution on $[-1,1]^d$ is
\begin{align*}
D^2(\xi; F_\asin)
& = \left(1+\frac{2}{\pi}-\frac{4}{\pi^2}\right)^d - \frac{2}{n}\sum_{i=1}^m n_i \prod_{j=1}^d \left[1+\frac{1}{\pi}+\frac{1}{2}|x_{ij}|-\frac{1}{\pi}\left(x_{ij}\arcsin(x_{ij})+\sqrt{1-x_{ij}^2} \right)\right]\nonumber \\
& \qquad + \frac{1}{n^2}\sum_{i,k=1}^m n_in_k\prod_{j=1}^d\left[1+\frac{1}{2}\left(|x_{ij}|+|x_{kj}|-|x_{ij}-x_{kj}| \right) \right].
\end{align*}

\bibliographystyle{asa}
\bibliography{GLM}

\end{document}